\documentclass[12pt,a4paper]{article}
\usepackage{amsmath,amsfonts,amssymb,amsthm, epsfig,
epstopdf, titling, url, array}
\usepackage[utf8]{inputenc}
\usepackage{mathtools}
\usepackage[english]{babel}
\usepackage{hyperref}
\newtheorem{thm}{Theorem}[section]
\newtheorem{lem}{Lemma}[section]
\newtheorem{cor}{Corollary}[section]
\newtheorem{prop}{Proposition}[section]
\numberwithin{equation}{section}

\def\C{\mathbb C}
\def\D{\mathbb D}

\def\R{I\!\!R}

\def\C{I\!\!\!\!C}

\def\D{I\!\!D}
\title{An integral representation for the resolvent kernel with magnetic fields on the hyperbolic plane and applications to time dependent Schr\"odinger equations}
\author{Mohamed Vall Ould Moustapha}

\begin{document}
\maketitle
\begin{abstract}

In this paper we give an integral representation for the resolvent kernels with uniform magnetic field on the hyperbolic plane, as applications of our results we solve explicitly two times dependent  Schr\"odinger equations with uniform magnetic field on the hyperbolic plane. 
\end{abstract}
 
\section{Introduction}
 
The main objective of this
paper is to give an integral  formula see Theorems \ref{Int-Rep1}  and \ref{Int-Rep2} for the Shwartz kernels $G_k(\lambda, w, w')$ of  the resolvent operators 
$\left({\cal D}_k+\lambda^2\right)^{-1}$,
with the operator ${\cal D}_{k} $ is the modified Schr\"odinger operator with uniform magnetic field on the  hyperbolic plane 
given by 
\begin{align}{\cal D}_{k}= {\cal L}_{k}^{\D}+k^2+\frac{1}{4},\end{align} 
where ${\cal L}_{k}^{\D}$ is the  Schr\"odinger operator with uniform magnetic field on the hyperbolic disc $\D$ given in (Ferapontov-Vesel \cite{FERAPONTOV-VESEL}) by
 \begin{align}{\cal L}_{k}^{\D}=(1-|w|^2)^2\frac{\partial^2}{\partial w \partial \overline{w}}+k(1-|w|^2)w\frac{\partial}{\partial w}& \nonumber\\
-k(1-|w|^2)\overline{w}\frac{\partial}{\partial \overline{w}}-k^2|w|^2.\end{align}

%%%%%%%%%%%%%%%%%%%%%%%%%%%%%%%%%%%%%%%%%%%%%%%%%%%%%%%%%%%%%%%%%%%%

The operators ${\cal L}_{k}^{\D}$   has a physical interpretation as being the Hamiltonian which governs a non relativistic charged particle
moving under the influence of the magnetic field of constant strength $|k|$, perpendicular to $\D$.   
The operators  ${\cal D}_{k}$ is
non-positive, definite
and it has an absolute continuous spectrum and a points spectrum if $|k|\ge 1/2$ see Boussejra -Intissar \cite{BOUSSEJRA-INTISSAR}.\\
For a recent work on this operator (see 
Ould Moustapha \cite{OULD MOUSTAPHA}) and the references therein.\\
%%%%%%%%%%%%%%%%%%%%%%%%%%%%%%%%%%%%%%%%%%%%%%%%%%%%
 For $k=0$, the operator  ${\cal D}_{0}$   reduces to the free Laplace-Beltrami operators on the hyperbolic plane.\\
The free resolvent kernel on the disc  model 
is given by 
\begin{align} G_0(\lambda, w, w')=
\frac{\Gamma(s)\Gamma(s)}{4\pi\Gamma(2s)}\cosh^{-2s}r(w, w')F\left(s, s , 2s,  \cosh^{-2}r(w, w')\right),
\end{align}
with $s=(1-i\lambda)/2$ and  $r(w, w')$  is  the distance.\\
The function $F(a, b, c, z)$ is the Gauss hypergeometric function defined by:
\begin{align}\label{Gauss}
  F(a, b, c, z)=\sum_{n=0}^{\infty}\frac{(a)_n(b)_n}{(c)_n n!}z^n,
  \quad |z|<1,
\end{align}
 $(a)_n$ is the Pochhamer symbol
  $(a)_n=\frac{\Gamma(a+n)}{\Gamma(a}$
and $\Gamma$ is the classical Euler function.
\section{Resolvent kernel with uniform magnetic field on the hyperbolic plane}
In this section we give explicit formulas for the resolvent kernel with uniform magnetic field  on the hyperbolic disc  $\D=\{w\in \C, |w|<1\}$ endowed with the metric $ds=2
\frac{|dw|}{(1-|w|^2)}$.\\  The associated group of motions is
\begin{align} G= SU(1, 1)=\left\{g= \left(
\begin{array}{cc}
A & \bar{B}\\
B & \bar{A}
\end{array}
\right) A, B\in \C: |A|^2-|B|^2=1\right\}.\end{align}
%%%%%%%%%%%%%%%%%%%%%%%%%%%%%%%%%%%%%%%%%%%%%%%%%%%%%%%%%%%%%%%%
 The distance $r(w, w')$ between two point $w, w' \in \D$ is given by
 \begin{align}\label{dist1}\cosh^2 (d(w, w')/2)=\frac{|1-w\overline{w'}|^2}{(1-|w|^2)(1-|w'|^2)}.\end{align}
We define here a unitary projective representation  $T^{k}$ of the group $G$ on the Hilbert space $L_\mu^2(\D)=L^2(\D, d\mu)$ by
\begin{align}T^{k}(g)f(z)=J(g^{-1}, z)f(g^{-1
}w),\end{align}
where the automorphic factor
\begin{align} J_k(g, w)=\left(\frac{\overline{Cw+D}}{Cw+D}\right)^{k}
\end{align}
 is defined, up to a mild ambiguity in the
 $k$ powers, which depends only on $G$ and $k$ and such ambiguity disappears if 
the group $SU(1, 1)$ is replaced by the  universal covering group $\widetilde{SU(1, 1)}$.\\
 Note that the factor $J_k$ satisfies, up to a factor of modulus one depending
only on $G$ and $k$, the chain rule
\begin{align}
J_k(g_1.g_2. z)=J_k(g_1, g_2.z)J_k(g_2, z).
\end{align}
In the sequel we use the following proposition
\begin{prop}
i) For $k$ real number, the modified Schr\"odinger operator with uniform magnetic field \ ${\cal D}_{k}$ is $T^{k}$ invariant, that is we have for every $g\in G$
\begin{align}\label{invar1}T^{k}(g){\cal D}_{k}={\cal D}_{k}T^{k}(g).\end{align}
  ii)For $w\in\D$ Set \begin{align}\label{gw}g_w=\left(
\begin{array}{cc}
\frac{1}{\sqrt{1-|w|^2}} &
\frac{ w}{\sqrt{1-|w|^2}}\\
\frac{\bar{w}}{\sqrt{1-|w|^2}} & 
\frac{1}{\sqrt{1-|w|^2}}
\end{array}
\right),\end{align} then  we have
$g_w\in G$ and  $g_w0=w$.\\
iii) For $f\in  L_\mu^2(\D)$ the following formulas hold
\begin{align} \label{Tkw}\left[T^{k}(g_{w'})f\right](w)=\left(\frac{1-\overline{w}w'}{1-w\overline{w'}}\right)^k f(g_{w'}^{-1}w).\end{align}
iv) For $f, \varphi\in  L_\mu^2(\D)$  we have
\begin{align} \label{Self}\int_{\D}\left[T^{k}(g)f]\right]
(w)\varphi(w)d\mu(w)=\int_{\D}f(w)[T^{-k}(g^{-1})\varphi](w)d\mu(w).
\end{align}
\end{prop}
\begin{proof} The proof of the part i) is contained in (Boussejra-Intissar \cite{BOUSSEJRA-INTISSAR}). The proof of the parts ii), iii)
and iv) are simple and are left as an exercise.
\end{proof}
\begin{lem} \label{lem1}
Let $u\in C^\infty (\D)$ be a radial function and let  $\Phi \in C^\infty (\R^+)$  such that
  $u(w) =v(r)$ with $ r = d(0, w)$ and $v(r)= \phi(y)$, with $y = \cosh^2(r/2)$ then we have\\
  i) ${\cal D}_{k} u(w) = l_{y}^k \Phi(y)$, with 
\begin{align}l^{k}_y \phi(y) = \left[y(y-1)\frac{d^2}{dy^2}+(2y-1)\frac{d}{dy}
 +\frac{k^2}{y}+\frac{1}{4}\right]\phi(y).\end{align}
 ii) Setting $\Phi(y)=y^{k}\Psi(y)$, then
$y^{-|k|}l_{y}^{|k|}y^{k}\psi(y)=J_k \psi(y)$, with
\begin{align}J_k =y(y-1)\frac{d^2}{dy^2}+[(2|k|+2)y-(2|k|+1)]\frac{d}{dy}+(|k|+1/2)^2.\end{align}
\end{lem}
\begin{proof} Using the geodesic polar coordinates,
$w=\tanh r/2\, \omega$, $r> 0$ and $\omega\in S^{1}$ we see that the radial part of the Magnetic Laplacian ${\cal D}_{k}$
 is given by
\begin{align}D_{k}=\frac{\partial^{2}}{\partial r^{2}}+\coth r \frac{\partial}{\partial r}+\frac{k^2 }{\cosh^2(r/2)}+\frac{1}{4}.\end{align}
Using the variables changes $y=\cosh^2(r/2)$,  we get the result of i). The part  ii) is simple and is left to the reader.
\end{proof}
\begin{prop}\label{prop2}
The Helmholtz equation with magnetic field on the hyperbolic disc model 
\begin{align}\label{Helmholtz}
\left({\cal D}_{k}+\lambda^2\right)u(z)=0,
\end{align}
has two linearly independent solutions
\begin{align}y^{|k|}F(s+|k|, 1-s+|k| , 1 ,  1-y),\ \ \ \  y^{-s}F(s-|k|, s+|k| , 2s,  y^{-1}).
\end{align}
where  $F(a, b, c, z)$ is the Gauss hypergeometric given in\eqref{Gauss}
 \end{prop}
 \begin{proof}
 From the Lemma \ref{lem1} the Helmholtz equation with magnetic field on the hyperbolic disc model 
\eqref{Helmholtz} 
 is equivalent to
 \begin{align}\{y(1-y)\frac{d^2}{dy^2}+[(2|k|+1)-(2|k|+2)y]\frac{d}{dy}-(|k|+1/2)^2+\lambda^2\}\psi=0.\end{align} 
 This equation is an hypergeometric equation (see Magnus et al. \cite{MAGNUS et al.} p.42-43) with parameters
 $a=|k|+1/2-i\lambda$, $b=|k|+1/2+i\lambda$ and $c=2|k|+1$ and the above solutions correspond respectively to the solutions
$w^{(1)}_1$ and  $w^{(\infty)}_1$, and the proof of Proposition \ref{prop2} is finished.
\end{proof}

%%%%%%%%%%%%%%%%%%%%%%%%%%%%%%%%%%%%%%%%%%%%%%%%%%%%%%%%%%%%%%%%%%%%%%%%%
\begin{lem}\label{lem2} Set
 \begin{align}\label{GD}G_k(s, r(0, w)=\frac{\Gamma(s-k)\Gamma(s+k)}{4\pi\Gamma(2s)}y^{-s}F\left(s-|k|, s+|k|, 2s, \cosh^{-2}r(0, w)\right).
\end{align} 
\begin{itemize} 
\item[i)] 
 $G(s, r)$ is analytic in 
$\lambda$ and $C^\infty$ in $r$ for $r>0.$ 
\item[ii)] $ G(s, r)\sim=-\frac{1}{4\pi}\ln\sinh^2r/2  ;  \ \  \ r \longrightarrow 0.$ 
\item[iii)] $ \frac{\partial  G(s, r)}{\partial r}=-\frac{1}{2\pi}\sinh^{-1}r/2,\ \ \  r \longrightarrow 0.$
\item[iv)] $ G(s, r)= \frac{\Gamma(s-k)\Gamma(s+k)}{4\pi\Gamma(2s)}\sinh^{-2s}r + O(\sinh^{-2s-1}r/2), as\, ,\ r \longrightarrow \infty .$  
\end{itemize}
\end{lem}
The proof of this lemma uses essentially the properties of hypergeometric functions in (Magnus et al.\cite{MAGNUS et al.} p.44).
$$F(a, b, a+b, z)=\frac{\Gamma(a+b)}{\Gamma(a)\Gamma(b)}\sum_{n=0}^{\infty}\frac{(a)_n(b)_n}{(n!)^2}[2\psi(n+1)-\psi(a+n)-\psi(b+n)-\ln(1-z)(1-z)^n, $$ 
$\arg(1-z)<\pi, |1-z|<1$, and

$F(a, b, a+b-m, z)=\frac{\Gamma(a+b-m)\Gamma(m)}{\Gamma(a)\Gamma(b)}(1-z)^{-m} \sum_{n=0}^{m-1}\frac{(a-m)_n(b-m)_n}{(n!)(1-m)_n}(1-z)^n$ \\ $-\frac{(-1)^n\Gamma(a+b-m}{\Gamma(a-n)\Gamma(b-n)}\sum_{n=0}^{\infty}\frac{(a)_n(b)_n}{(n!)(n+m)!}[\ln(1-z)-\psi(n+1)-\psi(n+m+1)](1-z)^n, $ 
$\arg(1-z)<\pi, |1-z|<1.$

\begin{thm}\label{thm1} Let $G_k(\lambda, r)$ be function given in \eqref{GD},
then we have
\begin{align}
\left({-\cal D}_k-\lambda^2\right)G=\delta, 
\end{align}
where $\delta$ is the Dirac measure
\end{thm}
\begin{proof}
\begin{align}
\langle \left({\cal D}_k+\lambda^2\right)G, \varphi\rangle=
 \langle G, \left({\cal D}_k+\lambda^2\right)\varphi \rangle= \lim_{\epsilon\longrightarrow 0}I_{\epsilon}.
\end{align}
Set \begin{equation}
I_\epsilon=J_\epsilon+ K_\epsilon,
\end{equation}
with \begin{align}
J_\epsilon=\int_\epsilon^\infty G(\lambda, r)\frac{\partial}{\partial r}\sinh r \frac{\partial}{\partial r}\varphi^\#(r\omega)d r,                  
\end{align}
\begin{align}
K_\epsilon=\int_\epsilon^\infty G(\lambda, r) \left(\frac{k^2}{\cosh^2r}+\frac{1}{4}+\lambda^2\right)\sinh r\varphi^\#(r\omega)d r.                 
\end{align}
Performing two integrations by parts we obtain\\
 $J_{\epsilon}=-G(\lambda,\epsilon)\sinh \epsilon[\varphi^{\#}]'(\epsilon\omega)+ G'(\lambda, \epsilon) \sinh \epsilon\varphi^{\#}(\epsilon\omega)$
 \begin{align}+\int_{\epsilon}^{\infty} \frac{\partial}{\partial r}\sinh r \frac{\partial}{\partial r}G(\lambda, r)\varphi^{\#}(r\omega)dr,                  
\end{align}

$I_\epsilon=-G(\lambda,\epsilon)\sinh \epsilon\varphi'^{\#}(\epsilon\omega)+ G'(\lambda, \epsilon)
 \sinh \epsilon\varphi^{\#}(\epsilon\omega).$\\
Using Lemma \ref{lem2} we obtain
\begin{align}
\langle \left({\cal D}_k+\lambda^2\right)G, \varphi\rangle=
 \langle G, \left({\cal D}_k+\lambda^2\right)\varphi \rangle= \lim_{\epsilon\longrightarrow 0}I_{\epsilon}
=-\varphi(0).
\end{align}
\end{proof}

%%%%%%%%%%%%%%%%%%%%%%%%%%%%%%%%%%%%%%%%%%%%%%%%%%%%%%%%%%%%%%%%%%%%%%%%%%%%%%%%%%%%%%%%%%%%%%%%%%%%%
%%%%%%%%%%%%%%%%%%%%%%%%%%%%%%%%%%%%%%%%%%%%%%%%%%%%%%%%%%%%%%%%%%%%%%%%%%%%%%%%%%%%%%%%%%%%%%%%%%%%%
\begin{thm}\label{thm2} 
Set $G_k(s, w, w')= T^ {k}(g_{w'})[G_{k}(s, r(0, w))]$ where $G_{k}(t, r(0, w))$ is given by
\eqref{GD} and $T^ {k}(g_{w'})$ is as in \eqref{Tkw} \, then we have: 
\begin{align}\label{G}G_k(s, w, w')=
\frac{\Gamma(s-k)\Gamma(s+k)}{4\pi\Gamma(2s)}\left(\frac{1-w\overline{w'}}{1-\overline{w}w'}\right)^{k}y^{-s}F\left(s-|k|, s+|k|, 2s, y^{-1}\right),
\end{align}
%%%%%%%%%%%%%%%%%%%%%%%%%%%%%%%%%%%%%%%%%%%%%%%%%%%%%%%%%%%%%%%%%%%%%%%%%%%%%%%%%%%%%%%%%%%%%%%%%%%
with $y=\cosh^{-2}(r(w, w')/2)$.
%%%%%%%%%%%%%%%%%%%%%%%%%%%%%%%%%%%%%%%%%%%%%%%%%%%%%%%%%%%%%%%%%%%%%%%%%%%%%%%%%%%%%%%%%%%
The following formulas hold
\begin{align}\left({-\cal D}_{k}^w -\lambda^2\right)G_{k}(s, w,  w')=\delta_{w'},\end{align}
\begin{align}\left({-\cal D}_{-k}^{w'} -\lambda^2\right)G_{k}(s, w,  w')=\delta_{w},\end{align}
where ${\cal D}_{k}^w$ is the modified Schr\"odinger operator with magnetic field with respect to $w$.
\end{thm}
\begin{proof}
The first formula is consequences of \eqref{GD} and \eqref{Tkw}.\\
For the last two results, Using the formula \eqref{invar1} and \eqref{Self} we can write 
\begin{align}\left({-\cal D}_{k}^w -\lambda^2\right)G_{k}(s, w,  w')& =\left({-\cal D}_{k}^w -\lambda^2\right)T^ {k}(g_w')[G_{k}(s, r(0, w))]\end{align}
\begin{align}=T^ {k}(g_w')\left({-\cal D}_{k}^w -\lambda^2\right)[G_{k}(s, r(0, w))]
=T^ {k}(g_w')\delta=\delta_{w'}.\end{align}

\begin{align}\left({-\cal D}_{k}^w -\lambda^2\right)G_{k}(s, w,  w')& =\left({-\cal D}_{k}^w -\lambda^2\right)T^ {k}(g_w')[G_{k}(s, r(0, w))]\end{align}
\begin{align}=T^ {k}(g_w')\left({-\cal D}_{k}^w -\lambda^2\right)[G_{k}(s, r(0, w))]
=T^ {k}(g_w')\delta=\delta_{w'}.\end{align}
\end{proof}

%%%%%%%%%%%%%%%%%%%%%%%%%%%%%%%%%%%%%%%%%%%%%%%%%%%%%%%%%%%%%%%%%%%%%%%%%%%%%%%%%%%%%%%%%%%%%%%
%%%%%%%%%%%%%%%%%%%%%%%%%%%%%%%%%%%%%%%%%%%%%%%%%%%%%%%%%%%%%%%%%%%%%%%%%%%%%%%%%%

%%%%%%%%%%%%%%%%%%%%%%%%%%%%%%%%%%%%%%%%%%%%%%%%%%%%%%%%%%%%%%%%%%%%%%%%%%%%%%%%%%%%%%%%%%%%%%%%%

%%%%%%%%%%%%%%%%%%%%%%%%%%%%%%%%%%%%%%%%%%%%%%%%%%%%%%%%%%%%%%%%%%%%%%%%%%%%%%%%%%%%%%%%%%%%%%%%%%%%%%

\begin{lem}\label{keylemma}
\begin{itemize}
\item{i)} For $\Re \mu > 0$\ and $\Re \nu > 0$, we have
\begin{align} \int_x^z (y-x)^{\mu-1}(z-y)^{\nu-1}y^{a-b-\mu}dy=\frac{\Gamma(\mu)\Gamma(\nu)}{\Gamma(\nu+\mu)}(y-x)^{\nu+\mu-1}\nonumber\\
F(b-a-\mu, \mu, \nu+\mu, 1-\frac{z}{x}).\end{align}
\item{ii)} For $ x > 1$ , $\Re \mu > 0$\ $\Re \nu > 0$\ and $\Re a > 0$  and $\Re b > 0,$
 \begin{align} \Gamma(a)\Gamma(b)x^{-
a+\mu}F(a, b, c, x^{-1})=\frac{\Gamma(a+\nu)\Gamma(b+\mu)}{\Gamma(\nu+\mu)}\int_x^{+\infty} I(x, z)z^{-a-\nu}\nonumber\\
F(a+\nu, b+\mu, c, z^{-1})\,dz \end{align}  
with
\begin{align} I(x,z)=(z-x)^{\nu+\mu-1}F(b-a-\mu, \mu, \nu+\mu, 1-\frac{z}{x}).
\end{align}
\end{itemize}
\end{lem}
\begin{proof} To see i) set $y-x=s$ we can write\\
$\int_x^z (y-x)^{\mu-1}(z-y)^{\nu-1}y^{a-b-\mu}dy=
\int_0^{z-x}s^{\mu-1}(z-x-s)^{\nu-1}(x+s)^{a-b-\mu}ds$.\\
Set $s=(z-x)t$, then we can write
\\ $\int_x^z (y-x)^{\mu-1}(z-y)^{\nu-1}y^{a-b-\mu}dy=$\\
$(z-x)^{\mu+\nu-1} x^{a-b-\mu}\int_0^1 t^{\mu-1}(1-t)^{\nu-1}(1-(\frac{x-z}{x})t)^{a-b-\mu}dt,$\\
using the integral representation see(Magnus et al.\cite{MAGNUS et al.}, p.54).
\begin{align}\label{Gauss2}F(a, b, c, z)=\frac{\Gamma(c)}{\Gamma(b)\Gamma(c-b)}\int_0^1 t^{b-1}(1-t)^{c-b-1}(1-zt)^{-a}dt\end{align}
we obtain the result of i).\\
 To see ii) we use the following formula twice (Intissar et al. \cite{INTISSAR et al.})
\begin{align} \label{key1} \Gamma(b)x^{-
b}F(a, b, c, x^{-1})={\Gamma(b+\mu)\over \Gamma(\mu)}\int_x^{+\infty}y^{-b-\mu}(y-x)^{\mu-
1}\nonumber\\ F(a, b+\mu, c, y^{-1})\,dy,\end{align} 
with for $ x > 1$ , $\Re \mu > 0$\ and $\Re b > 0$, 
to obtain
\begin{align} \Gamma(b)x^{-
b}F(a, b, c, x^{-1})=\frac{\Gamma(a+\nu)\Gamma(b+\mu)}{\Gamma(a)\Gamma(\nu)\Gamma(\mu)}\nonumber\\
\int_x^{+\infty} \int_y^{+\infty} (y-x)^{\mu-1}(z-y)^{\nu-1}y^{a-b-\mu}dy z^{-a-\nu}F(a+\nu, b+\mu, c, z^{-1})\,dz, \end{align}
that is
\begin{align} \Gamma(a)\Gamma(b)x^{-
b}F(a, b, c, x^{-1})=\frac{\Gamma(a+\nu)\Gamma(b+\mu)}{\Gamma(\nu)\Gamma(\mu)}\nonumber\\
\int_x^{+\infty} \left[\int_x^z (y-x)^{\mu-1}(z-y)^{\nu-1}y^{a-b-\mu}dy\right]z^{-a-\nu}F(a+\nu, b+\mu, c, z^{-1})\,dz. \end{align}
Using i) we arrive at the result ii) and the proof of the lemma is finished.

 \begin{thm}\label{Int-Rep1}  Let $G_k(r, \lambda)$ \ be the resolvent kernel for the 
Schr\"odinger operator with magnetic potential on hyperbolic plane,  for $\Im \lambda > -1/2$, the following equality holds \\ 
\begin{align}G_k(\lambda, r)=\int_r^{+\infty} W_k(r, \rho)\frac{e^{-i\rho \lambda}}{2i\lambda}d\rho, \end{align}
with
\begin{align}W_k(r,\rho)=\frac{1}{2\pi}(\cosh^2\rho/2-\cosh^2r/2)^{-1/2}\nonumber \\ F\left(|k|, -|k|, 1/2, 1-\frac{\cosh^2\rho/2}{\cosh^2r/2}\right)
\end{align}
where $F(a, b, c, z)$ is the Gauss hypergeometric function in \eqref{Gauss}.
\end{thm}
\begin{proof}
We use ii) of Lemma \ref{keylemma} with $a=s-|k|, b=s+|k|,\mu=|k|, \nu=1/2-|k|, x=\cosh^2r/2, z=\cosh^2\rho/2$
as well as the formula (Magnus et al.\cite{MAGNUS et al.}, p.39),
\begin{align}   F\left(a+1, a+1/2, 2a+1, \frac{1}{\cosh^2 z}\right)=e^{-2a z}\coth z(2\cosh z)^{2a}.\end{align}
\end{proof}
%%%%%%%%%%%%%%%%%%%%%%%%%%%%%%%%%%%%%%%%%%%%%%%%%%%%%%%%%%%%%%%%%%%%%%%%%%%%%%%%%%%%%%%%%%%%%%%%%%%%%%%%

Note that for $k$ integer or a half of an integer  by (Magnus et al.\cite{MAGNUS et al.}, p.39) 
 $T_n(1-2x)=F(-n, n,  \frac{1}{2}, x),$ 
and we can write
\begin{align}\label{chebichev}W_k(t, r)=\frac{1}{2\pi}\left(\cosh^2(t/2)-
\cosh^2(r/2))\right)_+^{-1/2}\times \nonumber\\ T_{2 |k|}\left(\frac{\cosh (t/2)}{\cosh (r /2)}\right),\end{align}
where $T_{2k}(x)$ are the Chebichev polynomials of the first kind.\\

%%%%%%%%%%%%%%%%%%%%%%%%%%%%%%%%%%%%%%%%%%%%%%%%%%%%%%%%%%%%%%%%%

%%%%%%%%%%%%%%%%%%%%%%%%%%%%%%%%%%%%%%%%%%%%%%%%%%%%%%%%%%%%%%%%%%%%%%%%%%%%%%%%%%%%%%%%%%%

\end{proof}

%%%%%%%%%%%%%%%%%%%%%%%%%%%%%%%%%%%%%%%%%%%%%%%%%%%%%%%%%%%%%%%%%%

 \begin{thm}\label{Int-Rep2}  Let $G_k(\lambda, w, w')$ be the resolvent kernel with magnetic potential on the disc model of the hyperbolic plane, then  we have \\ 
\begin{align}\label{Key-Key}G_k(\lambda, w, w')=\int_r^{+\infty}W_k(r(w, w'), \rho)\frac{e^{i\lambda \rho}}{2i\mu}d\rho, \end{align}
with\\
$W _k(r,\rho)=\frac{1}{2\pi}\left(\frac{1-w\overline{w'}}{1-\overline{w}w'}\right)^{k}\times$ 
\begin{align}\label{Key-Key-Key} (\cosh^2\rho/2-\cosh^2r/2)^{-1/2}F\left(|k|,-|k|, 1/2,1-\frac{\cosh^2\rho/2}{\cosh^2r/2}\right).\end{align}

\end{thm}
%%%%%%%%%%%%%%%%%%%%%%%%%%%%%%%%%%%%%%%%%%%%%%%%%%%%%%%%%%%%%%%%%%%%%%%%%%%%%%%%%%%%%%%%%%%
\section{Applications}
In this section using our results,  we solve explicitly the following two times dependent Schr\"odinger equations with uniform magnetic fields on the hyperbolic plane, called respectively the wave and the heat equations with magnetic field on the hyperbolic plane.

\begin{align}\label{Cauchy-Wave} \left \{\begin{array}{cc}{\cal D}_{k} u(t, w)=\frac{\partial^2}{\partial
t^2}u(t, w)
, (t, w)\in \R^\ast_+\times \D \\ u(0, w)=0, u_t(0, w)=u_1(w), u_1\in
C^\infty_0(\D)\end{array}
\right., \end{align}
and
\begin{align}\label{Cauchy-Heat} \left \{\begin{array}{cc}{\cal D}_{k} v(t, w)=\frac{\partial}{\partial
t}v(t, w)
, (t, w)\in \R^\ast_+\times \D \\  v(0, w)=v_0,  v_0\in
C^\infty_0(\D)\end{array}
\right., \end{align}
%%%%%%%%%%%%%%%%%%%%%%%%%%%%%%%%%%%%%%%%%%%%%%%%%%%%%%%%%%%%%%%%%%%%%%%%%%%%%%%%%%
%%%%%%%%%%%%%%%%%%%%%%%%%%%%%%%%%%%%%%%%%
\begin{cor}  The Schwartz integral  kernel of the wave operator $\frac{\sin t \sqrt{-{\cal D}}} {\sqrt{-{\cal D}}}$  that solves the wave Cauchy problem\eqref{Cauchy-Wave} with uniform magnetic field is given by\\
$W_k(t, w, w')=\frac{1}{2\pi}\left(\frac{1-w\overline{w'}}{1-\overline{w}w'}\right)^{k}\times $  \begin{align}  (\cosh^2 t/2-\cosh^2r/2)^{-1/2}{}_2F_1\left(|k|,-|k|, 1/2,1-\frac{\cosh^2 t/2}{\cosh^2r(w, w')/2}\right),\end{align}
 with ${}_2F_1$ is the Gauss hypergeometric function \eqref{Gauss}.
\end{cor}
\begin{proof}
$$(a^2+y^2)^{-1}=\int_0^\infty e^{- a x}\,\frac{\sin x y} {y}\, dx,$$
with $y=\sqrt{-{\cal D}}$ and $a=-i\lambda$ we can write
$$(-\lambda^2-{\cal D})^{-1}=\int_0^\infty e^{ i\lambda  t}\,\frac{\sin t \sqrt{-{\cal D}}} {\sqrt{-{\cal D}}}\,dt,$$
comparing with the formula \eqref{Key-Key} we have the result of the Corollary.
\end{proof}
Note that the above corollary agrees with the results obtained in Ould Moustapha \cite{OULD MOUSTAPHA} see also Intissar-Ould Moustapha \cite{INTISSAR-OULD MOUSTAPHA}.
%%%%%%%%%%%%%%%%%%%%%%%%%%%%%%%%%%%%%%%%%%%%%%%%%%%%%%%%%%%%%%%%%%
%%%%%%%%%%%%%%%%%%%%%%%%%%%%%%%%%%%%%%%%%%%%%%%%%%%%%%%%%%%%%%%%%%%%
%%%%%%%%%%%%%%%%%%%%%%%%%%%%%%%%%%%%%%%%%%%%%%%%%%%%%%%%%%%%%%%%%%%%%%%%%%%%%%%%%%%%%%%%%%
%%%%%%%%%%%%%%%%%%%%%%%%%%%%%%%%%%%%%%%%%%%%%%%%%%%%%%%%%%%%%%%%%%%%%%%%%%%%%%%%%%%%%%%%%
\begin{cor} The Schwartz integral  kernel of the heat operator $e^{t \cal D}$  that solves the heat Cauchy problem\eqref{Cauchy-Heat}\ with uniform magnetic potential on hyperbolic plane is given by:
\begin{align}\label{heat kernel}H_k(t, w, w')=\int_r^\infty\frac{ e^{-b^2/4t}}{(4\pi t)^{3/2}}W_k(b, w, w')b db,\end{align}
with
\begin{align}W_k(b, w, w')=\left(\frac{1-\overline{w}w'}{1-w\overline{w'}}\right)^k(\cosh^2(b/2) - 
\cosh^2 d(w,w')/2)_+^{-1/2}\nonumber\\ F\left(|k|, -|k|, 1/2, 1-\frac{\cosh^2 b/2}{\cosh^2r/2}\right).\end{align}
\end{cor}
%%%%%%%%%%%%%%%%%%%%%%%%%%%%%%%%%%%%%%%%%%%%%%%%%%%%%%%%%%%%%%%%%%%%%%%%%%%%%%%%%%%%%%%%%%%
\begin{proof}
The resolvent kernel and the heat kernel are related by the Laplace and Laplace inverse transforms as
\begin{align} G_k(\lambda, w, w')=\int_0^\infty e^{-\lambda t}H_k(t, w, w')dt,
\end{align}
with
\begin{align} H_k(t, w, w')=\frac{-1}{2 i\pi}\int_{c-i\infty}^{c+i\infty} e^{\lambda t}G_k(\lambda, w, w')d\lambda.
\end{align}
By combining the formula \eqref{Key-Key} and the formula Prudnikov et al.  \cite{PRUDNIKOV et al.} p.52 $L^{-1}e^{-a\sqrt{p}}(x)=\frac{a}{\sqrt{4\pi x^3}}e^{-\frac{a^2}{4 x}},$
${\cal R}e p >0$ and ${\cal R}e a^2 >0$. 
By Fubini theorem we arrive at the formulas  \eqref{heat kernel} and the proof of the Corollary is finished.
\end{proof}
%%%%%%%%%%%%%%%%%%%%%%%%%%%%%%%%%%%%%%%%%%%%%%%%%%%%%%%%%%%%%%%%%%%%%%%%%%%%%%%%%%%%%%%%%%

  M.V. Ould Moustapha, \textsc{Department of Mathematic,
 College of Arts and Sciences-Gurayat,
 Jouf University-Kingdom of Saudi Arabia }.\\
\textsc{Facult\'e des Sciences et Techniques
Universit\'e de  Nouakchott Al-asriya,
Nouakchott-Mauritanie.}\\
  \textit{E-mail address}: \texttt{mohamedvall.ouldmoustapha230@gmail.com}
  %%%%%%%%%%%%%%%%%%%%%%%%%%%%%%%%%%%%%%%%%%%%%%%%%%%%%%%%%%%%%%%%%%%%%%%%%%%%%
  \end{document}